\documentclass[graybox]{svmult}

\usepackage{mathptmx}       
\usepackage{amsmath} 
\usepackage{helvet}         
\usepackage{courier}        
\usepackage{type1cm}        
 \usepackage{url}             
\usepackage{makeidx}         
\usepackage{array,multirow,graphicx}      
\usepackage{multicol}        
\usepackage[bottom]{footmisc}

\makeindex             

\begin{document}

\title*{Cycle-centrality in economic and biological networks} 
\titlerunning{Cycle-centrality in complex networks} 

\author{Pierre-Louis Giscard and Richard C. Wilson}
\authorrunning{P.-L. Giscard and R. C. Wilson}

\institute{Pierre-Louis Giscard \at Department of Computer Science, University of York, Deramore Lane, Heslington, York, YO10 5GH, United Kingdom, \email{pierre-louis.giscard@york.ac.uk}
\and Richard C. Wilson \at Department of Computer Science, University of York, Deramore Lane, Heslington, York, YO10 5GH, United Kingdom, \email{richard.wilson@york.ac.uk}}
%
%
\maketitle

\abstract*{Networks are versatile representations of the interactions between entities in complex systems. Cycles on such networks represent feedback processes which play a central role in system dynamics. 
In this work, we introduce a measure of the importance of any individual cycle, as the fraction of the total information flow of the network passing through the cycle. This measure is computationally cheap, numerically well-conditioned, induces a centrality measure on arbitrary subgraphs and reduces to the eigenvector centrality on vertices. We demonstrate that this measure accurately reflects the impact of events on strategic ensembles of economic sectors, notably in the US economy. As a second example, we show that in the protein-interaction network of the plant \textit{Arabidopsis thaliana}, a model based on cycle-centrality better accounts for pathogen activity than the state-of-art one. This translates into pathogen-targeted-proteins being concentrated in a small number of triads with high cycle-centrality. Algorithms to compute the centralities of selected cycles or subgraphs are available on the \textsc{Matlab} FileExchange}

\vspace{-35mm}
\abstract{Networks are versatile representations of the interactions between entities in complex systems. Cycles on such networks represent feedback processes which play a central role in system dynamics. In this work, we introduce a measure of the importance of any individual cycle, as the fraction of the total information flow of the network passing through the cycle. This measure is computationally cheap, numerically well-conditioned, induces a centrality measure on arbitrary subgraphs and reduces to the eigenvector centrality on vertices. We demonstrate that this measure accurately reflects the impact of events on strategic ensembles of economic sectors, notably in the US economy. As a second example, we show that in the protein-interaction network of the plant \textit{Arabidopsis thaliana}, a model based on cycle-centrality better accounts for pathogen activity than the state-of-art one. This translates into pathogen-targeted-proteins being concentrated in a small number of triads with high cycle-centrality. Algorithms for computing the centrality of cycles and subgraphs are available for download.}

\vspace{-3.5mm}
\section{Introduction}
\vspace{-3mm}
Networks, that is collections of nodes together with sets of edges linking some of these nodes, naturally encode relations (the edges) between entities (the nodes). The trajectories on a network, called walks, represent the dynamical processes of the system of entities. Networks and walks play a ubiquitous role across many domains, from economy to defence through biology and physics, where graphical models are essential tools to master the interactions and dynamics of complex systems.

Recent research on networks has slowly progressed from questions directly concerning individual entities, to questions regarding the dynamics of the system, from the local to the global scale. Already over the course of the development of vertex-centralities, i.e. measures of the importance of individual nodes, it became clear that vertex-neighborhoods, subgraphs and motifs were of paramount importance to understand the evolution of real networks \cite{Milo2002,Yeger2004}. For example, in a recent study of the propagation of economic shocks in input-output networks, Alatriste Contreras and Fagiolo concluded that ``the systemic
importance of industrial sectors should not be evaluated only
by looking at their economic size [i.e. properties of individual vertices], but also at their position and embeddedness in the complex fabric of input-output relations'' \cite{Conteras2014}. In a biological context, 
Estrada and Rodr\'iguez-Vel\'azquez showed that protein-lethality in \textit{Saccharomyces cerevisiae} was better accounted for by an analysis of the subgraphs to which a protein belongs in the protein-protein interaction network (PPI) rather than by its degree \cite{Estrada2005}. In another study, Mukthar \textit{et al.} showed that while a number of the proteins of the plant \textit{Arabidopsis thaliana} under attack by pathogens were high degree nodes (hubs) in the plant PPI, dozens of these proteins were ``targeted significantly more often [...] than expected given their respective degrees''. Following a thorough statistical analysis of these results, they concluded that protein-targeting by pathogens ``cannot be explained merely by the high connectivity of those target [proteins]'' \cite{Mukhtar2011}, hence calling for further studies on the network environment of these proteins. In addition, it is also well known  that in PPIs, certain small subgraphs of protein interactions, called motifs, are over-represented as compared to what one might expect from random networks \cite{Oltvai2002}. These motifs are believed to perform crucial roles in emergent biological functions \cite{Wuchty2003}, such as the formation of protein complexes, functions which are not readily apparent at the level of single proteins \cite{Hartwell1999, Oltvai2002}.

In spite of all of these observations, 
much attention is still devoted to individual nodes when exploring the dynamics and properties of complex networks. This is possibly because the versatility, ease of implementation and easy to grasp definition of many vertex centralities is lacking an equivalent at the cycle or subgraph level. It is a central objective of this work to remediate to this situation.

We introduce a centrality measure for individual cycles based on the premise that a cycle is central if it intersects an important proportion of all the information flows on the network. In concrete applications, these flows represent actual dynamical processes, that is sequences of interactions between discrete entities, such as wealth exchanges between economic actors or successions of protein reactions in a living organism. This premise provides a clear meaning for the centrality as well as a contextual framework within which to appraise its results. Mathematically, the flow-based formulation leads to a rigorous and unique definition for the cycle-centrality. Computationally, it costs no more to calculate than existing centrality measures on vertices. Numerically, it is well conditioned, always providing a result between 0 and 1.  
Finally, the measure is versatile for it induces a centrality measures on subgraphs and reduces to the eigenvector-centrality on vertices. Algorithms to compute the centralities of  cycles and subgraphs are available on the \textsc{Matlab} FileExchange \cite{MatlabFiles}.

\vspace{-6.5mm}
\section{Centrality measure: theory $\&$ motivations}
\vspace{-4mm}
The measure of cycle-centrality we propose is rooted in recent advances in the algebraic combinatorics of walks on graphs. In this work we only define the few concepts from this background that are necessary to comprehend the centrality measure. 

\vspace{-8mm}
\subsection{Notation}
\vspace{-5mm}
Throughout this article, we consider a finite network $G = (\mathcal{V} ;\mathcal{E})$, $N=|\mathcal{V}|$, $M=|\mathcal{E}|$, which may be weighted and directed. The adjacency matrix of $G$ is denoted $\mathsf{A}_G$ or simply $\mathsf{A}$. If $G$ is weighted then the entry $\mathsf{A}_{ij}$ is the weight of the edge $e_{ij}$ from $i$ to $j$ if this edge exists, and 0 otherwise.

A \textit{walk} $w$ of length $\ell(w)$ from $v_i$ to $v_j$ on $G$ is a sequence $w = e_{i i_1} e_{i_1 i_2} \cdots e_{i_{\ell-1} j}$ of $\ell$ contiguous edges. The walk $w$ is \textit{open} if $i \neq j$ and \textit{closed} (that is a cycle) otherwise.
A \textit{cycle}, also known in the literature under the names \textit{loop}, \textit{simple cycle}, \textit{elementary circuit} and \textit{self-avoiding polygon}, is a closed walk $w = e_{i i_1} e_{i_1 i_2} \cdots e_{i_{\ell-1} i}$ which does not cross the same vertex twice, that is, the indices $i,i_1,\hdots,i_{\ell-1}$ are all different.

\vspace{-8mm}
\subsection{Definition of the centrality measure}
\vspace{-3mm}
The basic observation underlying our proposed centrality measure for cycles is that structurally, a cycle should be important if it is visited by many walks on the network. Combinatorially, the problem of counting all the walks visiting at least one vertex of a cycle is the graph-theoretic equivalent of counting the integer multiples of a prime number. Indeed, walks, it turns out, obey a semi-commutative extension of number theory in which cycles play the role of the primes. This framework, which is presented elsewhere \cite{Giscard2016}, notably provides an exact formula for the total number of closed walks on the graph which intersect the cycle $\gamma$. Asymptotically, this formula produces a single real number between 0 and 1, a fraction, representing the proportion of cycles intersecting the cycle $\gamma$. It is this number that we propose to use as a marker of structural cycle-importance in networks.
\begin{definition}[cycle centrality]
Let $G$ be a possibly weighted (di)graph, and let $\lambda$ be its maximum eigenvalue. Let $\mathsf{A}$ be the adjacency matrix of $G$, including weights if any. For any cycle $\gamma$, let $\mathsf{A}_{G\backslash \gamma}$ be the adjacency matrix of the graph $G$ where all vertices visited by $\gamma$ and the edges adjacent to them have been removed. Then we define the centrality $c(\gamma)$ of the cycle $\gamma$ as
\vspace{-2mm}
$$
c(\gamma):=\det\left(\mathsf{I}-\frac{1}{\lambda}\mathsf{A}_{G\backslash \gamma}\right).
\vspace{-2mm}
$$
\end{definition}
As outlined in the introduction to this section, the centrality $c(\gamma)$ has a precise combinatorial meaning underpinning its role as a measure of cycle importance. Rigorously we have:
\vspace{-1mm}
\begin{proposition}\label{cMeaning}
Let $G$ be a (di)graph with adjacency matrix $\mathsf{A}$ and let $\gamma$ be a cycle on $G$. Then the total number $n_\gamma(k)$ of cycles of length $k$ on $G$ intersecting the cycle $\gamma$ is asymptotically equal to
\vspace{-2mm}
$$
n_\gamma(k)\sim c(\gamma)\!\left(\frac{1}{\det(\mathsf{I}-z\mathsf{A})}\right)\![k],~~\text{as}~~k\to\infty,
\vspace{-2mm}
$$
where $\big(1/\det(\mathsf{I}-z\mathsf{A})\big)[k]$ stands for the coefficient of order $k$ in the series $1/\det(\mathsf{I}-z\mathsf{A})$.
\end{proposition}

\vspace{-4.5mm}
\begin{remark}
If $G$ is weighted then the above Proposition remains true but $n_\gamma(k)$ now designates the total weight of all the cycles of length $k$ on $G$ intersecting $\gamma$. Recall that the weight of a cycle is the product of the weights of the edges it traverses. The weights of different cycles are simply added together in $n_\gamma(k)$.\end{remark}
\begin{proof}[Qualitative proof of Proposition~\ref{cMeaning}]
The full rigorous proof of Proposition~\ref{cMeaning} is very long and will be provided in an extended version of this work. It relies on a semi-commutative extension of the Brun sieve from number theory, which provides the asymptotic result used here as well as an exact  expansion for $n_\gamma(k)$, of which $c(\gamma)$ is only the first term. Here we present a simple qualitative argument explaining the form of $c(\gamma)$ based on a result by X. G. Viennot concerning the combinatorics of heaps of pieces \cite{Viennot1989}. In the context of walks on graphs, Viennot's result indicates the following:
\begin{lemma}[Viennot (1986)]
Let $\gamma$ be a cycle on a finite graph $G$ and let $\mathcal{W}_\gamma$ be the set of closed walks intersecting $\gamma$. Then the ordinary generating series of all walks $w\in \mathcal{W}_\gamma$ is
\vspace{-2mm}
$$
\sum_{w\in\mathcal{W}_p}z^{\ell(w)} = \frac{\det(\mathsf{I}-z\mathsf{A}_{G\backslash \gamma})}{\det(\mathsf{I}-z\mathsf{A})}z^{\ell(\gamma)}.
\vspace{-1mm}
$$
\end{lemma}
Now let $1/\det(\mathsf{I}-z\mathsf{A})=\sum_{n=0}^{\infty} a_n z^n$ and $\det(\mathsf{I}-z\mathsf{A}_{G\backslash \gamma})=\sum_{n=0}^{N-\ell(\gamma)}x_n z^n$. If we expand the ratio of determinants from Viennot's lemma, the coefficient of order $k$ in the expansion then reads
\vspace{-4mm}
$$
a_k x_0 +a_{k-1} x_1 + a_{k-2} x_2+\cdots +a_0 x_k = \sum_{i=0}^k a_{k-i}\,x_i.
\vspace{-2mm}
$$
We remark that since the determinant $\det(\mathsf{I}-z\mathsf{A})$ is a polynomial in the eigenvalues of the graph $G$, asymptotically, the coefficient of order $k$ of its inverse $1/\det(\mathsf{I}-z\mathsf{A})$ should grow as $\lambda^k$. Taking $a_k=\lambda^k$ for all $k$, it would follow that $a_{k-i}=a_k\lambda^{-i}$ and 
\vspace{-2mm}
$$
\sum_{i=0}^k a_{k-i}\,x_i = a_k \sum_{i=0}^k \frac{x_i}{\lambda^i}.
\vspace{-1mm}
$$
In the situation where $k\geq N-\ell(\gamma)$, no term is missing from the sum on the right hand side, i.e. $\sum_{i=0}^k \frac{x_i}{\lambda^i}= \det(\mathsf{I}-\frac{1}{\lambda}\mathsf{A}_{G\backslash \gamma})$, \textit{qualitatively} explaining the form of $c(\gamma)$. It is remarkable that this form is unchanged by fully rigorous arguments in which the (incorrect) assumption that $a_k=\lambda^k$ is relaxed.\qed
\end{proof}

We can further confirm the meaning of $c(\gamma)$ by noting that the series $1/\det(\mathsf{I}-z\mathsf{A})$ itself has a combinatorial meaning: it counts multi-ensemble of walks, known as hikes \cite{Giscard2016}. Then $c(\gamma)$ is the (weighted) fraction of such multi-ensembles which are closed walks intersecting $\gamma$. In other words, $c(\gamma)$ is the proportion of the total information flow of the network that passes through $\gamma$.
A corollary of these observations is that the cycle-centrality satisfies a highly desirable property for such measures:
\begin{proposition}
Let $G$ be a (weighted di)graph with non-negative edge weights and let $\gamma$ be a cycle on $G$.
Then
$$
0\leq c(\gamma)\leq 1.
$$
\end{proposition}
\begin{proof}
The result $c(\gamma)\leq 1$ follows immediately from Proposition~\ref{cMeaning}, by observing that there are necessarily less walks than multi-ensembles of walks (hikes), i.e. $n_{\gamma}(k)\leq \big(1/\det(\mathsf{I}-z\mathsf{A})\big)[k]$. This continues to hold true on positively weighted graphs, where the total weight carried by walks is necessarily less than that carried by hikes. 
Positivity of $c(\gamma)$ follows from the positivity of both $n_{\gamma}(k)$ and $\big(1/\det(\mathsf{I}-z\mathsf{A})\big)[k]$, itself guaranteed by the positivity of the edge weights.\qed
\end{proof}

%
%

\vspace{-9mm}
\subsection{Extension to arbitrary subgraphs}
\vspace{-4mm}
The cycle-centrality measure $c(\gamma)$, although rooted in the combinatorics of cycles, naturally extends to a  centrality measure $c(H)$ for induced subgraphs $H\prec G$, which quantifies the (weighted) proportion of closed walks, i.e. dynamical processes, intersecting the subgraph $H$.
\begin{definition}[Induced subgraph centrality]
Let $G$ be a (weighted di)graph, and let $\lambda$ be its maximum eigenvalue. Let $\mathsf{A}$ be the adjacency matrix of $G$, including weights if any. Let $H\prec G$ be an induced subgraph of $G$ and let $\mathsf{A}_{G\backslash H}$ be the adjacency matrix of the graph $G$ where all vertices of $H$ and the edges adjacent to them have been removed. Then we define the induced subgraph centrality as
\vspace{-2mm}
$$
c(H):=\det\left(\mathsf{I}-\frac{1}{\lambda}\mathsf{A}_{G\backslash H}\right).
\vspace{-1mm}
$$
\end{definition}
\begin{proposition}\label{cMeaning2}
Let $G$ be a (weighted di)graph with adjacency matrix $\mathsf{A}$ and let $H\prec G$ be an induced subgraph of $G$. Then the total number (weight) $n_H(k)$ of cycles of length $k$ on $G$ intercepting the subgraph $H$ is 
\vspace{-2mm}
$$
n_H(k)\sim c(H)\!\left(\frac{1}{\det(\mathsf{I}-z\mathsf{A})}\right)\![k],~~\text{as}~~k\to\infty.
\vspace{-1mm}
$$
\end{proposition}
The proof of this result is similar to that of Proposition~\ref{cMeaning}. Furthermore it  holds that if all weights are non-negative, then for any induced subgraph $H$ of $G$, $0\leq c(H)\leq 1$.

\vspace{-7.5mm}
\subsection{Recovering the eigenvector vertex centrality}
\vspace{-3.5mm}
The problem of quantifying the importance of individual nodes in networks has a long history of research, which has led to well established measures such as the degree \cite{Albert1999},  exponential \cite{Estrada2005}, resolvent \cite{Katz1953} and eigenvector centralities \cite{Bonacich1972}, the latter playing plays a central role in network analysis, notably through the PageRank algorithm \cite{Bryan2006}. 
Recall that the centrality of vertex $i$ in the first three measures is the (weighted) degree of vertex $i$; and the $i$th entries of $e^{\mathsf{A}}\mathbf{1}$ and $(\mathsf{I}-\alpha\mathsf{A})^{-1}\mathbf{1}$, respectively. In these expressions $\mathbf{1}$ is the column vector full of ones and $0\leq\alpha< 1/\lambda$ is the Katz parameter. 
The last measure, the eigenvector centrality, here denoted $\text{eig}(i)$, is the value of the $i$th entry of the eigenvector corresponding to the largest eigenvalue of the graph. 

These measures stem from the idea that a vertex is important if it is the starting point of many closed walks. Naively summing over all such walks to define a centrality leads to a divergent sum however. This problem was resolved through two strategies: i) giving walks of length $\ell$ an additional weight of $\alpha^\ell$, with $\alpha<1/\lambda$ to guarantee convergence of the sum. This yields the resolvent centrality. Or ii) giving such walks a weight of $1/\ell!$, once again to guarantee convergence, an approach which yields the exponential centrality. Both approaches lend smaller weights to longer walks and are thus mostly sensitive to the close neighbourhoods of vertices.
A third point of view is based on the observation that, on Markov chains, the $i$th entry of the dominant eigenvector is the probability of walking from $i$ to itself in the stationary distribution. The eigenvector centrality is based on an extension of this observation.  

In the context of cycles and induced subgraphs, a natural way to define a consistent vertex centrality measure is to set it to be the centrality $c(i)$ of the singleton subgraph containing only the vertex $i$.\footnote{Since the cycle centrality and its extension to subgraphs are consistent, $c(i)$ is also equal to the cycle centrality of a self-loop $i\to i$ from vertex $i$ to itself.} Immediately then $c(i)$ is the asymptotic proportion of closed walks passing through $i$ on $G$ and a measure of the importance of this vertex. This centrality is essentially the same as the eigenvector centrality:
\vspace{-1mm}
\begin{proposition}\label{EigCentr}
Let $G$ be a (weighted) graph with adjacency matrix $\mathsf{A}$ and largest 
eigenvalue $\lambda$. Let $\eta:=\lim_{z\to 1/\lambda}(1-
\lambda z)^{-1}\det(\mathsf{I}-z\mathsf{A})$.
Then 
\vspace{-2mm}
$$
c(i) = \eta\, \text{eig}(i)^2.
$$
\vspace{-8mm}
\end{proposition}
\begin{proof}
Let $\mathsf{R}(z):=(\mathsf{I}-z\mathsf{A})^{-1}$ be the resolvent of the graph. 
The adjugate formula indicates that 
\vspace{-3mm}
$$
\text{Adj}(\mathsf{I}-z \mathsf{A})_{ii}=\det(\mathsf{I}-z \mathsf{A})\mathsf{R}(z)_{ii}=\,\det(\mathsf{I}-z \mathsf{A}_{G\backslash i}).
\vspace{-1mm}
$$
Hence $\lim_{z\to1/\lambda}\text{Adj}(\mathsf{I}-z \mathsf{A})_{ii}=c(i)$. But since $\lambda$ is the largest eigenvalue of $G$, the adjugate in this formula tends to the projector $\mathsf{P}_\lambda$ onto the corresponding (dominant) eigenvector. More precisely and assuming that the conditions for the Perron-Frobenius Theorem hold
then $\lim_{z\to1/\lambda}\text{Adj}(\mathsf{I}-z \mathsf{A})_{ii} = \eta \big(\mathsf{P}_\lambda\big)_{\!ii}$\,, from which the result follows.\qed 
\end{proof}

\vspace{-3mm}
\begin{remark}The idea of using network flows to measure the importance of vertices was first proposed by Freeman and coworkers \cite{Freeman1977, Freeman1991}.
In spite of conceptual similarities with the cycle-centrality introduced here, these measures are genuinely different. The flow-betweenness centrality of a vertex is defined either as the number of shortest simple paths \cite{Freeman1977} or of all simple paths \cite{Freeman1991} passing through a given vertex. In this context, a simple path is a walk which is not allowed to visit any vertex more than once. As a consequence, the flow-betweenness is computationally difficult to obtain, the problem of counting simple paths being $\#$P-complete \cite{Valiant1979}. 
\end{remark}

Further mathematical properties of the cycle-centrality will be presented in an extended version of this work.

\vspace{-8mm}
\subsection{Computational cost}
\vspace{-3.5mm}
Computationally speaking, since $c(\gamma)$ involves a determinant, it costs $O\big((N-|\gamma|)^{3}\big)$ operations to calculate, where $N-|\gamma|$ is the number of vertices of the graph $G\backslash\gamma$. In practice however, $c(\gamma)$ can be approximated, even on very large networks, by retaining only a set $\{\mu_1,\cdots,\mu_q\}$ of dominant eigenvalues of $\mathsf{A}_{G\backslash \gamma}$ (or $\mathsf{A}_{G\backslash H}$), yielding $c(\gamma)\simeq\prod_{i=1}^q(1-\mu_i/\lambda).$ Convergence of this approximation 
can be tested by increasing the number $q$ of retained eigenvalues. 
Algorithms for the calculation of $c(\gamma)$ are available for download, see \cite{MatlabFiles}.

In addition, a problem encountered by \textit{any} cycle- or subgraph-based centrality is that it requires some knowledge of the cycles or subgraphs whose importance is to be measured. For exemple, to rank all connected induced subgraphs on $\ell$ vertices by centrality value, regardless of what this centrality is defined to be, one must first find them. This step costs $O(N\Delta^\ell)$ operations, $\Delta$ being the maximum degree of the graph.  
There is therefore no doubt that cycle- and subgraph-centralities will, overall, be computationally more expensive to calculate than vertex-centralities. What we argue here, is that the performances of the resulting models justify the additional costs of the analysis. 

\vspace{-7.5mm}
\section{Economic networks}\label{EcoSection}
\vspace{-3mm}
\begin{figure}[t!]
\vspace{-2mm}
\centering
\includegraphics[width=.8\linewidth]{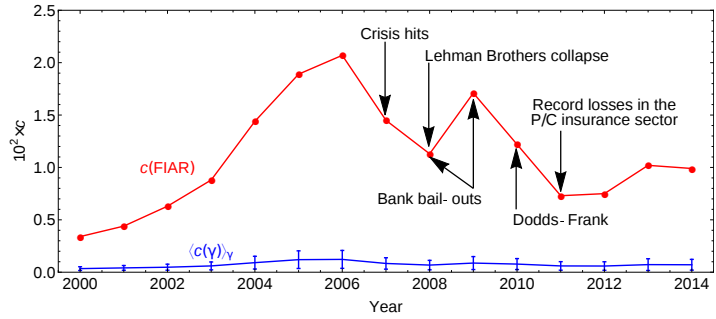}
\vspace{-2mm}
\caption{\label{FIAR} Temporal evolution of the cycle-centrality of the FIAR clique (red circles) during the 2000$-$2014 period compared to the average cycle-centrality of all four-vertices cycles $\langle c(\gamma)\rangle$ (blue) with error-bars indicating the standard deviation around the average. Important events over the period are indicated. 
\vspace{-5mm}}
\end{figure} 
\vspace{-1mm}
In order to test the viability of the cycle-centrality defined in the preceding section, we studied the economy of the United-States, United-Kingdom, Germany and France  over the period 2000$-$2014 using the World Input-Output Database (2016 release) \cite{Timmer2015, Timmer2016, WIO2016web}. The data provides the flow of capital on a yearly basis between 55 sectors of the economy, yielding a weighted directed  macro-economic graphical model of the evolution of each country during 15 full years. 

We calculated the cycle-centrality of all 1,485 edges, 26,235 triangles, 341,055 squares and 3,478,761 pentagons on each network for each year of the interval 2000$-$2014 \cite{MatlabFiles}. This task took a total of circa 1hour per country on a 3.1~GHz Intel Core i7 MacBook Pro.

\vspace{-8mm}
\subsection{Preliminary results}
\vspace{-4mm}
In the case of the US, the most important edge on average over the period 2000$-$2014 was found to involve the real-estate and insurance sectors, whilst all triangles and squares whose centrality was within $40\%$ of the maximum observed centrality involved these two sectors and/or the financial industry. This is different for other countries, e.g. the dominant set of actors of the German economy was composed of the ``Manufacture of motor vehicles and trailers'', ``Real estate activities'' and ``Administrative and support service activities''. The French economy saw the most capital flowing through the ``Electricity, gas, steam and air conditioning supply'', ``Construction'' and ``Legal and accounting activities; activities of head offices and management consultancy'' sectors, with cycles involving the ``Manufacture of food, beverages and tobacco products'' following closely behind. This sector was also important in the US economy, being present in most of the dominant squares and pentagons. Overall, these results confirm that the cycle-centrality functions well as an indicator of the importance of groups of agents in dynamical processes on complex networks, in this case dominant sets of sectors ranked in terms of capital flows. 

\vspace{-7mm}
\subsection{Case study: finance-insurance-real estate in the US economy}
\vspace{-3mm}
\subsubsection{Evolution of intercepted capital flow}
\vspace{-2mm}
Of great interest to the study of the recent economic history of the United-States is the role played by the finance, insurance and real estate sectors. We thus selected the following four sectors for further study:
\begin{itemize}
\setlength\itemsep{-.1em}
\item ``\textbf{F}inancial service activities, except insurance and pension funding'';
\item ``\textbf{I}nsurance, reinsurance and pension funding, except compulsory social security'';
\item ``\textbf{A}ctivities auxiliary to financial services and insurance activities'';
\item ``\textbf{R}eal estate activities''.
\end{itemize}
\vspace{-1mm}
These four sectors form a clique, here called ``FIAR''. The cycle-centrality $c(\text{FIAR})$ evaluates the weight and frequency with which wealth exchange within the US economy passed through the FIAR clique over the 15 years period from 2000 to 2014. The results are shown on Fig.~(\ref{FIAR}) and correlate with the main events surrounding the 2007-2009 crisis, showing the effects of bank bail-outs and the introduction of the Dodds-Frank act. This legislation is seen to have contained and stabilised the importance of the FIAR clique in the US economy.  To these observations, we can perhaps add the Financial Services Modernization Act of 1999 which repealed parts of the Glass$-$Steagall Act and which explains the subsequent exponential increase of $c(\text{FIAR})$ over the period 2000$-$2006 \cite{Greenwood2013}.

The central role played the FIAR sectors is perhaps best  summarised by a single number: the time-average $\langle c(\text{FIAR})\rangle_\tau$ is nearly 15 times higher than the cycle and time average $\langle \langle c(\gamma)\rangle_\gamma\rangle_\tau$, where the cycle averaging $\langle .\rangle_\gamma$ is effected over all four-vertices cycles of the graph. This means the FIAR clique intersected on average 15 times more capital flow every year of the 2000$-$2014 period than the average ensemble of four economic sectors.

 \vspace{-4mm}
\subsubsection{Comparison with alternative centralities}
\vspace{-2mm}
In order to contrast the performance of $c(\gamma)$ with those of existing measures, a simple approach consists in defining cycle-centralities from the sum of the vertex centralities of the individual vertices visited by the cycles. Although this strategy has been criticised \cite{Everett1999}, it provides insights into the information content of the various measures. We thus define $\Sigma_{CS}(\gamma)$, $\Sigma_{R}(\gamma)$ and $\Sigma_{eig}(\gamma)$, the sums of the exponential, resolvent and eigenvector centralities of the vertices visited by $\gamma$, respectively.\footnote{
In the case of $\Sigma_{CS}(\gamma)$, we had to introduce a regularisation parameter $r$ such that $e^{\mathsf{A}/r}$ converges in Matlab and $\Sigma_{CS}(\gamma)$ could be computed.
We then verified that the relative 
variations of $\Sigma_{CS}(\gamma)$ were \emph{qualitatively} independent 
from $r$.
}

\begin{figure}[t!]
\vspace{-2mm}
\centering
\includegraphics[width=.8\linewidth]{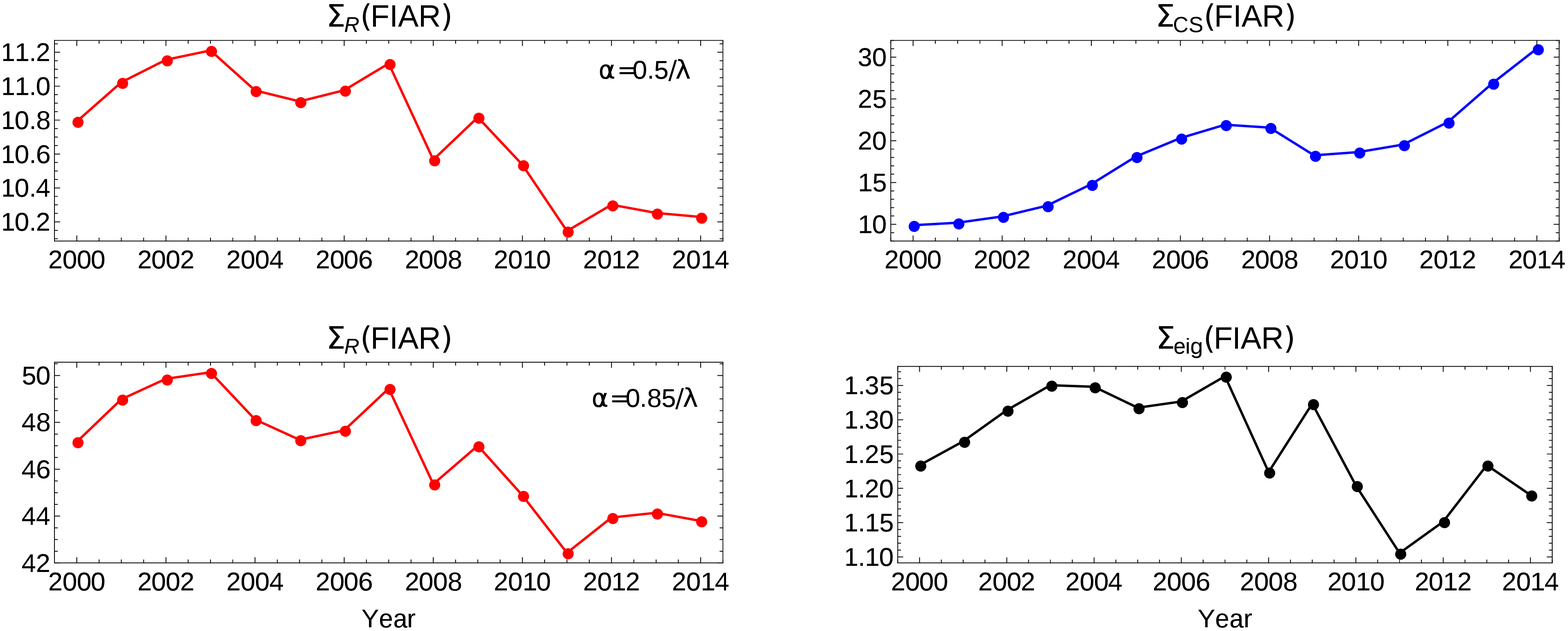}
\vspace{-1mm}
\caption{\label{OC} Left column: evolution of the resolvent-based centrality measure for the FIAR clique with two of the most commonly-used values for the Katz parameter: top $\alpha=0.5/\lambda$; bottom $\alpha=0.85/\lambda$. Top right: temporal evolution of the exponential-based centrality measure for the FIAR clique with a regularisation parameter of $10^5$. Bottom right: evolution of the eigenvector-based centrality measure for the FIAR clique.}
\vspace{-5mm}
\end{figure} 

The results for the FIAR clique are shown on Fig.~(\ref{OC}) and indicate the comparative failure of these  approaches. 
For example, according to $\Sigma_{R}(\text{FIAR})$, and regardless of the Katz parameter employed, the FIAR clique underwent a massive downturn between 2003 and 2007, when all economic indicators show that this period was one of unprecedented growth for the FIAR sectors \cite{Greenwood2013, Treasury2016, Economist2017}.

The eigenvector-centrality-based measure $\Sigma_{eig}(\text{FIAR})$ does not fare much better. According to it, one should believe that: 1) the  importance of the FIAR sectors in 2014 was slightly lower than in 2000; and 2) the centrality of the FIAR sectors inexplicably peaked in the year 2003 only to reach the same level over the year 2007. Both conclusions 1) and 2) are in contradiction with economic studies on the subject, especially concerning the year 2007 when the crisis saw the collapse of much of the FIAR sectors \cite{Greenwood2013, Economist2017}.

Finally, from the $\Sigma_{CS}(\text{FIAR})$ centrality, one should believe that by 2014,  the combined importance of the finance, insurance and real estate sectors was much higher\footnote{How much higher exactly depends on the regularisation parameter. The smaller $r$, the higher the ratio of the centralities between the years 2014 and 2006. For its smallest value guaranteeing convergence $r\sim 10^3$, the ratio is a totally improbable $4\times 10^{23}$.} than its maximum pre-crisis level in late 2006 $-$ early 2007. Yet, it is known that the housing market was more than 20$\%$ lower in real terms in 2014 than at its peak in late 2006 \cite{Economist2017}; and that the net-income of the insurance industry, in particular the Property and Casualty subsector most connected to the real-estate industry, was comparable in 2014 to its 2006 level after record losses in 2011, see \cite{Treasury2016} p. 32.\\[-1em] 
 
In addition, the particular values taken by the resolvent-, eigenvector- and exponential-based centrality measures are rather difficult to interpret since they are not immediately related to quantities of real-world significance. 
At the opposite, the cycle-centrality $c(\gamma)$ is the proportion of the total capital flow of the economy that passes through the cycle $\gamma$. 
As a consequence, the results of an analysis using this cycle-centrality are easy to 
grasp and interpret, facilitating their appraisal with respect to external sources of information.
We must conclude from these observations and those of the preceding paragraphs that the insights gained from this analysis 
are not easily replicated by other centrality measures.

\vspace{-7mm}
\section{Protein targeting in plant-pathogen interactions}\label{BioSection}
\vspace{-4mm}
We now turn to a biological context and consider the protein-protein interaction network (PPI) obtained by Mukhtar \textit{et al.} in a landmark study of plant-pathogens interactions between the plant \textit{Arabidopsis thaliana}, the bacterium \textit{Pseudomonas syringae} and the oomycete \textit{Hyaloperonospora arabidopsidis}. The network, 
comprises 3,148 interactions between 926 proteins, of which 170 are known to participate in plant immunity and 137 are targeted by effectors from one or both pathogens \cite{Mukhtar2011}. 

Before the original study of \cite{Mukhtar2011} it was already expected that the pathogens would target those proteins which are the most important to the plant \cite{Landry2011}, i.e. that most of the pathogen targets should be high-degree nodes (hubs) in the plant PPI. Here we call this hypothesis the degree-based model of protein-targeting. The model posits a positive correlation between protein-targeting and the degree-centrality of the proteins. Mukhtar \textit{et al.} confirmed such a correlation, showing it to be statistically significant, yet also observed shortfalls of the model, such as numerous low-degree targets and hubs targeted by few pathogen-effectors, if at all. Nonetheless, the degree-based model is the best available vertex-based model as replacing the degree by another vertex-centrality degrades model-performances, see Table~(\ref{CentTarget}). 

In their seminal study, Mukhtar \textit{et al.} also showed that highly connected proteins  tend to be involved in immune interactions \cite{Mukhtar2011, Landry2011}. Furthermore, subsequent biological  studies, notably into oomycetes, have shown that pathogen effectors are potent stimulants of immune activity in \textit{Arabidopsis thaliana} \cite{Oome2014, Fawke2015}. 
Consequently, we might expect the PPI to comprise small protein motifs involving not only a pathogen target, but also one or more interactions with an immune protein, interactions which may be stimulated by the activity of the pathogen on the target, and an accompanying central protein.\footnote{\label{CentralityFootnote}Here the centrality of a protein is understood to be its eigenvector centrality since, by Proposition~\ref{EigCentr}, this is the measure induced by the cycle-centrality on vertices.}
If we now hypothesise that pathogens primarily aim at disrupting a sizeable proportion of sequences of protein reactions in the host, then the motifs mentioned above should  have high cycle-centrality. This is because in the context of PPIs the cycle-centrality of a motif measures the fraction of sequences of protein interactions intercepted by the motif. In other words, pathogen-targets should primarily be found in triads with dominant cycle-centrality involving at least one target, one or more central proteins, and one or more immune interactions. 
%
%


To test this model, which we call the dominant-triad model, we calculated the cycle-centrality of all 113,398 triads of 
proteins in the PPI using \cite{MatlabFiles}, which took $27$ min on the aforementioned computer, most of which was spent finding the triads. We then selected those triads 
involving at least one of the top two 2 proteins in terms of eigenvector centrality\footnotemark[4] (circa $2\%$ of all triads). These are AT5G08080 and AT5G22290 (in that order of centrality). The 
former likely belongs to a set of proteins involved in plant resistance against bacteria 
\cite{Kalde2007, TAIR1}, while the latter belongs to a family of a  transcription factors 
with a role in stress responses. More precisely, AT5G22290 negatively regulates 
flowering in response to stresses \cite{Li2010, TAIR2}. Remarkably AT5G08080 is not 
targeted at all by the pathogens, while AT5G22290 is targeted by a single effector in 
spite of being the most important hub of the plant PPI, with a degree of 222.\footnote{By 
contrast another protein, AT3G47620, is targeted by 29 effectors from both the 
bacterium and the oomycete yet has ``only'' degree 104 \cite{Mukhtar2011}.} 
Among the triads comprising AT5G08080 and/or AT5G22290, we classified as true positive those which involve at least one more target and at least one immune reaction. Finally, in order to compare the performances of the dominant-triad and degree-based models, we obtained the ROC curves for both. The results, presented on Fig.~(\ref{Triad}), clearly show the dominant-triad model out-performing the degree-based one of \cite{Mukhtar2011}. The performances of models based on the centralities $\Sigma_{CS}$, $\Sigma_R$ and $\Sigma_{eig}$ are reported in Table~(\ref{CentTarget}) for comparison.\footnote{Running the computations separately, each of these models would take the c. 30~min time to evaluate since the majority of this time is spent finding the triads.}

These results suggest that the hypothesis where pathogens select their targets to maximise the fraction of disrupted sequences of protein reactions better fits the observations than the hypothesis where they target high-degree nodes of the PPI. In particular, the model explains why hubs are not the only targets nor necessarily the most targeted proteins, as interactions with peripheral proteins in the immediate vicinity of a central protein are seemingly equally disruptive to the ensemble of sequences of reactions on the PPI. The performance of the dominant-triad model also underscores the remarkable efficiency of the plant immune response: as the ROC curve shows, nearly all triads with the highest cycle-centrality involving a pathogen target also involve an immune interaction. 
Taken together, these observations paint the picture of a PPI where two central proteins are immediately surrounded by numerous pathogen targets and a flurry of immune interactions. 
\begin{figure}[t!]
\begin{center}
\vspace{-4mm}
\includegraphics[width=.7\linewidth]{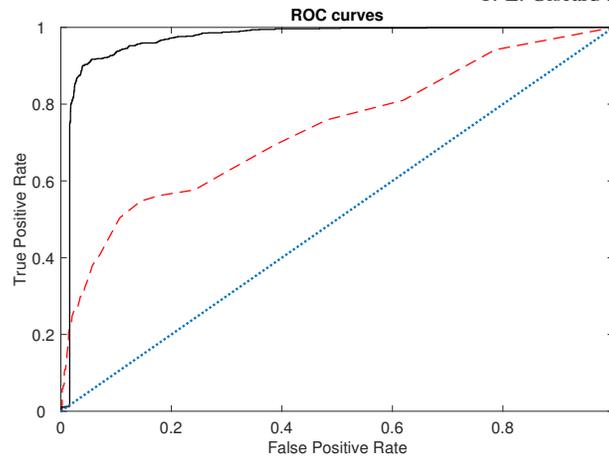}
\vspace{-3mm}
\end{center}
\caption{\label{Triad} Solid black line: ROC curve of the dominant-triad model in the plant-pathogen PPI of \textit{A. thaliana}, \textit{P. syringae} and \textit{H. arabidopsidis}. In this model, all triads involving AT5G08080 and/or AT5G22290 are ranked in descending order according to cycle-centrality. A true positive is a triad involving at least one more target and at least one immune reaction, while a false positive is a triad which does not meet both of these criteria. 
Red dashed line: ROC curve of the degree-based model proposed in \cite{Mukhtar2011}, where proteins are ranked in descending order according to their degree in the PPI. A true positive is a protein targeted by at least one pathogen effector. Dotted line: null-hypothesis model with random protein-targeting.}
\vspace{-2mm} 
\end{figure}

\vspace{-3.5mm}
\begin{table}
\centering
\textbf{Performances of protein-targeting models}
\begin{tabular}{clccc}
&Model & ROC AUC & Discrimination \\
\hline
 \parbox[t]{3mm}{\multirow{5}{*}{\rotatebox[origin=c]{90}{\scriptsize{Cycle-based}}}}&1. Dominant-triad $c(\gamma)$ & \textbf{0.97} & $\mathbf{0.47}$  \\
&2. $\Sigma_{R}$ ($\alpha=0.85/\lambda$) & 0.89 & 0.39  \\
&3. $\Sigma_{CS}$ & 0.88 & $0.38$  \\
&4. $\Sigma_{eig}$ & 0.87 & $0.37$  \\
&5. $\Sigma_{R}$ ($\alpha=0.5/\lambda$) & 0.85 & $0.35$  \\
\hline
 \parbox[t]{3mm}{\multirow{5}{*}{\rotatebox[origin=c]{90}{\scriptsize{Vertex-based}}}}&6. Degree centrality \cite{Mukhtar2011}& 0.73 & $0.23$  \\
&7. Resolvent centrality ($\alpha=0.5/\lambda$) & 0.28 & $0.22$ \\
&8. Resolvent centrality ($\alpha=0.85/\lambda$) & 0.28 & $0.22$ \\
&9. Exponential centrality & 0.60 & $0.10$ \\
&10. Eigenvector centrality & 0.41 & $0.09$ \\
\hline
\end{tabular}
\vspace{1mm}
\caption{\label{CentTarget} The ROC AUC is the area under the ROC curve. A perfect model making only correct predictions would have ROC AUC=1, while the null-hypothesis yields ROC AUC=0.5. The discrimination is the (absolute) area between the ROC curve and the null-hypothesis line. The crucial difference between vertex-based and cycle-based models is that the former attempt at directly identifying individual protein-targets, while the latter aim at identifying targeted triads.}
\vspace{-7mm} 
\end{table}

\vspace{-6mm}
\section{Conclusion}
\vspace{-4mm}
We have introduced a centrality measure for cycles on networks that quantifies the (weighted) fraction of information flow intersecting the cycle. This measure is computationally cheap to calculate, numerically well-conditioned, and extends both to arbitrary subgraphs and vertices, where it reduces to the eigenvector centrality. 
By studying the evolution of the US economy over the period 2000$-$2014, we have shown that the cycle-centrality correlates with major events impacting this economy and could potentially serve as an objective quantifier of the effect of crashes and novel legislations. In the biological context of plant-pathogens interactions, we have shown that a model where pathogens select their targets to maximise the number of sequences of protein reactions that they intercept better fits for the observations than a model where pathogens target high-degree proteins.

\vspace{-3.5mm}
\begin{acknowledgement}
We thank Paul Rochet of the Laboratoire Jean-Leray, Nantes, France, for stimulating discussions. P.-L. Giscard is grateful for the financial support from the Royal Commission for the Exhibition of 1851.
\end{acknowledgement}


\vspace{-11mm}
\begin{thebibliography}{99.}%
\vspace{-4.5mm}
\bibitem{Albert1999}
Albert R, Jeong H, Barab\'asi AL (1999) {Internet: Diameter of the World-Wide
  Web}.
\newblock {\em Nature} 401:130--131.

\bibitem{Bonacich1972}
Bonacich P (1972) {Factoring and weighting approaches to status scores and
  clique identification}.
\newblock {\em Journal of Mathematical Sociology} 2(1):113--120.


\bibitem{Bryan2006}
Bryan K, Leise T (2006) {The $\$25,000,000,000$ Eigenvector: The Linear Algebra
  behind Google}.
\newblock {\em SIAM Review} 48(3):569--581.

\bibitem{Conteras2014}
Contreras MGA, Fagiolo G (2014) {Propagation of economic shocks in input-output
  networks: A cross-country analysis}.
\newblock {\em Phys. Rev. E} 90:062812.

\bibitem{Estrada2005}
Estrada E, Rodr\'iguez-Vel\'azquez JA (2005) {Subgraph centrality in complex
  networks}.
\newblock {\em Physical Review E} 71:056103.

\bibitem{Everett1999}
Everett MG, Borgatti SP (1999) {The centrality of groups and classes}
\newblock {\em J. Math. Sociol.} 23(3):181--201

\bibitem{Fawke2015}
Fawke S, Doumane M, Schornack S (2015) {Oomycete Interactions with Plants:
  Infection Strategies and Resistance Principles}.
\newblock {\em Microbiol. Mol. Biol. Rev.} 79(3):263--280.

\bibitem{Freeman1977}
Freeman LC (1977) {A set of measures of centrality based on betweenness}.
\newblock {\em Sociometry} 40:35--41.

\bibitem{Freeman1991}
Freeman LC, Borgatti SP, White DR (1991) {Centrality in valued graphs: A
  measure of betweenness based on network flow}.
\newblock {\em Social Networks} 13(2):141--154.

\bibitem{Giscard2016}
Giscard PL, Rochet P (2017) Algebraic combinatorics on trace monoids: Extending
  number theory to walks on graphs.
\newblock {\em SIAM J. Discrete Math.} 31(2):1428--1453.

\bibitem{MatlabFiles}
Giscard PL, Wilson RC (2017) Algorithm to calculate the cycle-centrality of selected cycles or subgraphs: \url{https://mathworks.com/matlabcentral/fileexchange/64678}. Algorithm to calculate the centrality of all connected induced subgraphs of fixed size: \url{https://mathworks.com/matlabcentral/fileexchange/64677}.

\bibitem{Greenwood2013}
Greenwood R, Scharfstein D (2013) {The Growth of Finance}.
\newblock {\em Journal of Economic Perspectives} 27(2):3--28.

\bibitem{Hartwell1999}
Hartwell LH, Hopfield JJ, Leibler S, Murray AW (1999) {From molecular to modular cell biology}.
\newblock {\em Nature} 402(6761):C47--C52.

\bibitem{Kalde2007}
Kalde M, N\"{u}hse TS, Findlay K, Peck SC (2007) {The syntaxin SYP132
  contributes to plant resistance against bacteria and secretion of
  pathogenesis-related protein 1}.
\newblock {\em Proc. Natl. Acad. Sci. U.S.A.} 104(28):11850--11855.

\bibitem{Katz1953}
Katz L (1953) {A new status index derived from sociometric data analysis}.
\newblock {\em Psychometrika} 18:39--43.

\bibitem{Landry2011}
Landry C (2011) {A Cellular Roadmap for the Plant Kingdom}.
\newblock {\em Science} 333:532--533.


\bibitem{Li2010}
Li J, Zhang J, Wang X, Chen J (2010) {A membrane-tethered transcription factor
  ANAC089 negatively regulates floral initiation in \textit{Arabidopsis
  thaliana}}.
\newblock {\em Science China Life Sciences} 53(11):1299--1306.

\bibitem{Milo2002}
Milo R, et~al. (2002) {Network motifs: simple building blocks of complex
  networks}.
\newblock {\em Science} 298(5594):824--827.

\bibitem{Mukhtar2011}
Mukhtar MS, et~al. (2011) {Independently Evolved Virulence Effectors Converge
  onto Hubs in a Plant Immune System Network}.
\newblock {\em Science} 333(6042):596--601.

\bibitem{Oltvai2002}
Oltvai ZN, Barab\'asi AL (2002) {Life's complexity pyramid}.
\newblock {\em Science} 298(5594):763.

\bibitem{Oome2014}
Oome S, et~al. (2014) {Nep1-like proteins from three kingdoms of life act as a
  microbe-associated molecular pattern in Arabidopsis}.
\newblock {\em Proc. Natl. Acad. Sci. U.S.A.} 111(47):16955--16960.


\bibitem{Timmer2015}
Timmer MP, Dietzenbacher E, Los B, Stehrer R, de~Vries GJ (2015) {An
  Illustrated User Guide to the World Input-Output Database}.
\newblock {\em Review of International Economics} 23:575--605.

\bibitem{Timmer2016}
Timmer MP, Los B, Stehrer R, de~Vries GJ (2016) {An Anatomy of the Global Trade
  Slowdown based on the WIOD 2016 Release}.
\newblock {\em GGDC research memorandum} 162.

\bibitem{Valiant1979}
Valiant LG (1979) {The Complexity of Enumeration and Reliability Problems}.
\newblock {\em SIAM J. Comput.} 8(3):410--421.

\bibitem{Viennot1989}
Viennot GX (1989) Heaps of pieces, i: Basic definitions and combinatorial lemmas.
\newblock {\em Ann. N. Y. Acad. Sci.} 576:542--570.

\bibitem{Wuchty2003}
Wuchty S, Oltvai ZN, Barab\'asi AL (2003) {Evolutionary conservation of motif
  constituents in the yeast protein interaction network}.
\newblock {\em Nature Genetics} 35(2):176--179.

\bibitem{Yeger2004}
Yeger-Lotem E, et~al. (2004) {Network motifs in integrated cellular networks of
  transcription-regulation and protein-protein interaction}.
\newblock {\em Proc. Natl. Acad. Sci. U.S.A.} 101(16):5934--5939.

\bibitem{Treasury2016}
(2016) {Annual Report on the Insurance Industry}
  (\url{https://www.treasury.gov/initiatives/fio/reports-and-notices/Documents/2016_Annual_Report_FINAL.pdf}).
  
\bibitem{WIO2016web}
(2016) {World Input Output Database} (\url{http://www.wiod.org/database/niots16}).

\bibitem{Economist2017}
(2017) {American house prices: realty check}
  (\url{http://www.economist.com/blogs/graphicdetail/2016/08/daily-chart-20}).

\bibitem{TAIR1}
(2017) {Protein AT5G08080}
  (\url{https://www.arabidopsis.org/servlets/TairObject?name=AT5G08080&type=locus}).

\bibitem{TAIR2}
(2017) {Protein AT5G22290}
  (\url{https://www.arabidopsis.org/servlets/TairObject?name=AT5G22290&type=locus}).

\end{thebibliography}
\end{document}